\documentclass[11pt]{article}
\usepackage[T1]{fontenc}
\usepackage{lmodern}
\usepackage[margin=1in]{geometry}
\usepackage{amsmath,amssymb,amsthm}
\usepackage{booktabs}
\usepackage{microtype}
\usepackage[hidelinks]{hyperref}
\hypersetup{pdftitle={A symmetric conference matrix of order 86},
  pdfauthor={Harshit Verma}}
\numberwithin{equation}{section}
\newtheorem{theorem}{Theorem}[section]
\newtheorem{proposition}[theorem]{Proposition}
\newtheorem{lemma}[theorem]{Lemma}
\theoremstyle{remark}
\newtheorem{remark}[theorem]{Remark}
\newcommand{\one}{\mathbf{1}}
\newcommand{\Z}{\mathbb{Z}}
\newcommand{\R}{\mathbb{R}}
\newcommand{\T}{\mathsf{T}}
\DeclareMathOperator{\tr}{tr}
\DeclareMathOperator{\Aut}{Aut}

\title{A symmetric conference matrix of order 86}
\author{Harshit Verma\\[0.3em]
  \small Yale University\\
  \small\href{mailto:harshit.verma@yale.edu}{\texttt{harshit.verma@yale.edu}}}
\date{}

\begin{document}
\maketitle

\begin{abstract}
We construct a symmetric conference matrix of order $86$, the smallest
order for which existence was open. Equivalently, there exist a
conference graph with parameters $(85,42,20,21)$, a regular two-graph on
$86$ points, and a real equiangular tight frame of $86$ vectors in
$\R^{43}$. The matrix is a $12\times12$ array of $7\times7$ circulants
with two border rows, and it admits a group of automorphisms of order
$21$. Its $86$ descendants give exactly ten nonisomorphic strongly
regular graphs.
\end{abstract}

\section{Introduction}

A \emph{symmetric conference matrix} of order $n$ is a symmetric
$n\times n$ matrix $C$ with zero diagonal and entries $\pm1$ off the
diagonal such that $C^2=(n-1)I$. For $n>1$ such a matrix can exist only
if $n\equiv2\pmod4$ and $n-1$ is a sum of two squares. For $n\ge6$,
symmetric conference matrices of order $n$ correspond to conference graphs on
$n-1$ vertices, the strongly regular graphs with parameters
$\bigl(n-1,\frac{n-2}{2},\frac{n-6}{4},\frac{n-2}{4}\bigr)$: normalizing
a row of $C$ and deleting it gives such a graph (Section~\ref{sec:desc}),
and every conference graph arises in this way \cite[\S8.2]{BvM}.

The Paley graphs give examples whenever $n-1$ is a prime power
congruent to $1$ mod $4$. The smallest order $n>2$ allowed by the conditions
above with $n-1$ not a prime power is $n=46$, settled by Mathon
\cite{Mathon}, and Gritsenko settled $n=66$ \cite{Gritsenko}. The next
case, $n=86$, was the smallest open one \cite[\S2.3]{GoryainovEtAl}, and
the parameters $(85,42,20,21)$ are marked unknown in Brouwer's table
\cite{BrouwerTable}.

\begin{theorem}\label{thm:main}
There is a symmetric conference matrix of order $86$. Consequently,
there is a strongly regular graph with parameters $(85,42,20,21)$.
\end{theorem}

Section~\ref{sec:matrix} gives the matrix and reduces
Theorem~\ref{thm:main} to a finite identity. Section~\ref{sec:desc}
shows that the descendants of the matrix give exactly ten nonisomorphic
strongly regular graphs, and Section~\ref{sec:equiv} records the usual
equivalent forms of the result. Lemma~\ref{lem:id}, the invariance
under $h$ in Proposition~\ref{prop:group}, and the counts in
Table~\ref{tab:k4} are finite computations in the explicitly given matrix
and were checked in exact integer arithmetic. The matrix was found by a computer
search, which will be described in a later version; the proofs do not
depend on it.

\section{The matrix}\label{sec:matrix}

We write $I$, $J$ and $\one$ for identity matrices, all-ones matrices
and all-ones vectors, with sizes as subscripts when needed. Index rows
and columns by
\[
 X=\{\infty,\infty'\}\cup\{(i,a):0\le i<12,\ a\in\Z/7\Z\},
\]
ordered as $\infty,\infty',(0,0),\dots,(0,6),(1,0),\dots,(11,6)$, so
that $(i,a)$ is in position $2+7i+a$, counting from $0$. Let
\[
 t=(1,1,1,-1,-1,-1,1,1,1,-1,-1,-1)
\]
and
\begingroup\small\setlength{\arraycolsep}{3.2pt}
\begin{equation}\label{eq:M}
 M=\begin{pmatrix}
126&104&104&1&105&23&7&21&19&23&23&23\\
22&126&104&23&1&105&44&88&98&57&75&101\\
22&22&0&105&105&23&93&115&47&44&88&98\\
1&105&23&0&23&23&45&89&99&77&113&43\\
23&1&23&105&0&23&58&78&116&51&15&85\\
105&23&105&105&105&126&52&26&70&65&33&9\\
97&52&59&53&92&44&24&89&53&63&4&67\\
41&26&79&27&114&88&27&66&99&37&95&16\\
73&70&117&71&46&98&45&71&36&2&25&119\\
105&29&52&51&77&3&125&37&64&60&74&28\\
105&83&26&15&113&5&32&123&25&82&90&100\\
105&39&70&85&43&17&67&8&111&56&38&102
 \end{pmatrix}.
\end{equation}
\endgroup
For $0\le i,j<12$ and $r\in\Z/7\Z$, represented by $0\le r<7$, put
$s_{ij}(r)=-1$ if bit $r$ of $M_{ij}$ is set (bit $0$ being the least
significant) and $s_{ij}(r)=1$ otherwise, except that $s_{ii}(0)=0$.
For example, $M_{00}=126$ gives $s_{00}(0)=0$ and $s_{00}(r)=-1$ for
$r\ne0$.

Define $C$ to have zero diagonal and
\begin{equation}\label{eq:C}
 C_{\infty\infty'}=1,\qquad C_{\infty,(i,a)}=1,\qquad
 C_{\infty',(i,a)}=t_i,\qquad C_{(i,a),(j,b)}=s_{ij}(b-a),
\end{equation}
with $b-a$ taken mod $7$, and with the border entries extended
symmetrically. Let $B$ be the $84\times84$ submatrix on the coordinates
$(i,a)$; it is a $12\times12$ array of $7\times7$ circulants. With
$u=\one_{84}$ and $w=t\otimes\one_7$,
\[
 C=\begin{pmatrix}0&1&u^\T\\ 1&0&w^\T\\ u&w&B\end{pmatrix}.
\]

\begin{lemma}\label{lem:id}
In the ring $\Z[x]/(x^7-1)$ put $P_{ij}(x)=\sum_{r=0}^{6}s_{ij}(r)x^r$
and $\Phi(x)=1+x+\dots+x^6$, and let $P=(P_{ij})$ be the resulting
$12\times12$ matrix. Then
\begin{enumerate}
 \item[(a)] $P(x^{-1})^\T=P(x)$;
 \item[(b)] $P(1)\one=-t$ and $P(1)t=-\one$;
 \item[(c)] $P(x)^2=85I-(J+tt^\T)\,\Phi(x)$.
\end{enumerate}
\end{lemma}

\begin{proof}
Comparing coefficients, (a) says that $s_{ji}(-r)=s_{ij}(r)$, (b)
consists of $24$ row sums, and (c) says that
\[
 \sum_{k=0}^{11}\sum_{a\in\Z/7\Z}s_{ik}(a)\,s_{kj}(r-a)
 =85\,\delta_{ij}\delta_{r0}-1-t_it_j
\]
for all $i,j,r$. These are finitely many integer identities in the
entries of~\eqref{eq:M}, and a direct computation verifies all of
them.
\end{proof}

\begin{proof}[Proof of Theorem~\ref{thm:main}]
Let $Z$ be the $7\times7$ permutation matrix with $Z_{a,a+1}=1$, indices
mod $7$. Then $(Z^r)_{ab}=1$ exactly when $b-a=r$, so the $(i,j)$ block
of $B$ is $P_{ij}(Z)$. The map $x\mapsto Z$ is a ring homomorphism from
$\Z[x]/(x^7-1)$ to the $7\times7$ integer matrices, and $Z^\T=Z^{-1}$, $Z\one=\one$ and $\Phi(Z)=J_7$.
Lemma~\ref{lem:id} therefore gives
\[
 B^\T=B,\qquad Bu=-w,\qquad Bw=-u,\qquad B^2=85I-uu^\T-ww^\T,
\]
using $uu^\T=J_{12}\otimes J_7$ and $ww^\T=tt^\T\otimes J_7$ for the last
one. Since also $u^\T u=w^\T w=84$ and $u^\T w=0$, multiplying out the
block form of $C$ gives $C^2=85I$. The statement about graphs follows
from Proposition~\ref{prop:srg}.
\end{proof}

\section{Descendants and automorphisms}\label{sec:desc}

Let $\Omega$ be the set of $3$-subsets $\{x,y,z\}$ of $X$ with
$C_{xy}C_{yz}C_{zx}=-1$; this is the two-graph of $C$. Replacing $C$ by
$DCD$ for a diagonal $\pm1$ matrix $D$ does not change $\Omega$. For
$r\in X$, the \emph{descendant} $\Gamma_r$ is the graph on
$X\setminus\{r\}$ in which $x$ and $y$ are adjacent if
$\{r,x,y\}\in\Omega$. In matrix terms, switch $C$ so that row $r$ is $+1$
off the diagonal and delete row and column $r$; if $S$ is the remaining
$85\times85$ matrix, then $\Gamma_r$ has adjacency matrix $A=(J-I-S)/2$.

\begin{proposition}\label{prop:srg}
Every $\Gamma_r$ is strongly regular with parameters $(85,42,20,21)$.
\end{proposition}

\begin{proof}
The switched matrix still squares to $85I$ and its row $r$ is $+1$ off
the diagonal, so $S\one=0$ and $S^2=85I-J$. It follows that
$A\one=42\one$ and $A^2=21I-A+21J$, so adjacent vertices have $20$
common neighbors and nonadjacent vertices have $21$.
\end{proof}

Let $\sigma=(6\,7\,8)(9\,10\,11)$, acting on $\{0,\dots,11\}$, and define
permutations $g,h$ of $X$ that fix $\infty$ and $\infty'$ by
\[
 g(i,a)=(i,a+1),\qquad h(i,a)=(\sigma(i),2a).
\]

\begin{proposition}\label{prop:group}
The permutations $g$ and $h$ are automorphisms of $C$, that is,
$C_{g(x)g(y)}=C_{xy}$ for all $x,y\in X$, and likewise for $h$. They
generate a group $G\cong C_7\rtimes C_3$ of order $21$, whose orbits on
$X$ are $\{\infty\}$, $\{\infty'\}$, the six blocks $\{i\}\times\Z/7\Z$
with $i<6$, and the two sets $\{6,7,8\}\times\Z/7\Z$ and
$\{9,10,11\}\times\Z/7\Z$.
\end{proposition}

\begin{proof}
For $g$ this is the circulant structure of $B$. For $h$ it amounts to
$t_{\sigma(i)}=t_i$ and $s_{\sigma(i)\sigma(j)}(2r)=s_{ij}(r)$ for all
$i,j,r$, which is a finite check on~\eqref{eq:M}. Since $g^7=h^3=1$
and $hgh^{-1}=g^2$, every element of $G$ has the form $g^ah^b$, and these $21$ elements are
distinct because they act on the block $i=0$ as the maps
$z\mapsto2^bz+a$. The orbits are then clear.
\end{proof}

An automorphism $p$ of $C$ preserves $\Omega$, so it maps $\Gamma_r$
isomorphically onto $\Gamma_{p(r)}$. Hence descendants whose roots lie
in the same $G$-orbit are isomorphic.

\begin{theorem}\label{thm:ten}
The $86$ descendants $\Gamma_r$ fall into exactly ten isomorphism
classes, one for each $G$-orbit. In particular, there are at least ten
pairwise nonisomorphic strongly regular graphs with parameters
$(85,42,20,21)$.
\end{theorem}

\begin{proof}
By the preceding paragraph there are at most ten classes.
Table~\ref{tab:k4} gives the number $\kappa_4$ of $4$-cliques in
$\Gamma_r$ for one root in each orbit. The ten values are distinct, so
these ten graphs are pairwise nonisomorphic.
\end{proof}

\begin{table}[htbp]
\centering
\caption{Four-clique counts of the descendants, one root per $G$-orbit.
The index is the position of the root in the ordering of
Section~\ref{sec:matrix}.}\label{tab:k4}
\begin{tabular}{@{}lrrr@{}}
\toprule
Root $r$ & Index & Orbit size & $\kappa_4(\Gamma_r)$\\
\midrule
$\infty$  & 0  & 1  & 25900\\
$\infty'$ & 1  & 1  & 26152\\
$(0,0)$   & 2  & 7  & 26018\\
$(1,0)$   & 9  & 7  & 26182\\
$(2,0)$   & 16 & 7  & 26098\\
$(3,0)$   & 23 & 7  & 26084\\
$(4,0)$   & 30 & 7  & 26078\\
$(5,0)$   & 37 & 7  & 26096\\
$(6,0)$   & 44 & 21 & 26070\\
$(9,0)$   & 65 & 21 & 26086\\
\bottomrule
\end{tabular}
\end{table}

\begin{remark}
Since $\Gamma_\infty$ is the only descendant with $\kappa_4=25900$, every
automorphism of $\Omega$ fixes $\infty$, so $\Aut\Omega=\Aut\Gamma_\infty$.
Every automorphism of $\Gamma_\infty$ extends to $\Omega$ fixing $\infty$,
so it preserves the vertex colors $x\mapsto\kappa_4(\Gamma_x)$.
Exhaustive enumeration of the color-preserving automorphisms of
$\Gamma_\infty$ with \texttt{GraphMatcher} in NetworkX~3.2.1
\cite{NetworkX} gives $|\Aut\Gamma_\infty|=21$.
Hence $G$ is the full automorphism group of $\Omega$. In particular, it
is the full group of permutation automorphisms of $C$, and
$\Aut\Gamma_r$ is the stabilizer of $r$ in $G$.
\end{remark}

\section{Equivalent forms}\label{sec:equiv}

The following reformulations of Theorem~\ref{thm:main} are standard, so
we only sketch the arguments.

For distinct $x,y\in X$, the $84$ products $C_{xz}C_{zy}$ with
$z\ne x,y$ sum to $(C^2)_{xy}=0$. Hence exactly $42$ of the triples
$\{x,y,z\}$ lie in $\Omega$, and $\Omega$ is a regular two-graph on $86$
points \cite[\S8.10]{BvM}.

Since $C^2=85I$ and $\tr C=0$, the eigenvalues of $C$ are $\pm\sqrt{85}$,
each with multiplicity $43$. So $I+C/\sqrt{85}$ is positive semidefinite
of rank $43$ with unit diagonal. It is the Gram matrix of $86$ unit
vectors in $\R^{43}$ with all inner products $\pm1/\sqrt{85}$, which form
an equiangular tight frame. The coherence $1/\sqrt{85}$ meets the Welch
bound, so the $86$ lines they span are an optimal packing of $86$ lines in $\R^{43}$
\cite{StrohmerHeath}. This case is listed as open in
\cite[Table~1]{Waldron}; see also \cite{FickusMixon}.

Taking two copies $(x,\pm1)$ of each point of $X$ and joining
$(x,\varepsilon)$ to $(y,\eta)$ when $x\ne y$ and $\varepsilon\eta=C_{xy}$
gives the Taylor graph of $\Omega$ \cite{Taylor}, an antipodal
distance-regular graph on $172$ vertices with intersection array
$\{85,42,1;1,42,85\}$.

Finally, $C$ is normalized at $\infty$, so deleting row and column
$\infty$ leaves a symmetric matrix $S$ with zero diagonal, entries $\pm1$
elsewhere, $S\one=0$ and $S^2=85I-J$. More generally, if $S$ is a symmetric
matrix of order $q$ with zero diagonal, entries $\pm1$ elsewhere,
$S\one=0$ and $S^2=qI_q-J_q$, then
$S\otimes S+I_q\otimes J_q-J_q\otimes I_q$ has the same entry conditions,
has zero row sums, and squares to $q^2I_{q^2}-J_{q^2}$. Bordering it with
a row and column of ones, with zero at their intersection, gives a
symmetric conference matrix of order $q^2+1$ \cite[Proposition~8.2.2]{BvM}.
Starting from $q=85$, this gives symmetric conference matrices of every
order $85^{2^j}+1$ with $j\ge0$.


{\small
\begin{thebibliography}{10}

\bibitem{BrouwerTable}
A.~E. Brouwer,
\emph{Parameters of strongly regular graphs: 51--100 vertices},
online table, accessed October 6, 2026.
\url{https://aeb.win.tue.nl/graphs/srg/srgtab51-100.html}.

\bibitem{BvM}
A.~E. Brouwer and H. Van Maldeghem,
\emph{Strongly Regular Graphs},
Encyclopedia of Mathematics and its Applications, vol.~182,
Cambridge University Press, 2022.
\href{https://doi.org/10.1017/9781009057226}{doi:10.1017/9781009057226}.

\bibitem{FickusMixon}
M. Fickus and D.~G. Mixon,
Tables of the existence of equiangular tight frames,
preprint, 2015.
\href{https://arxiv.org/abs/1504.00253}{arXiv:1504.00253}.

\bibitem{GoryainovEtAl}
S. Goryainov, W.~H. Haemers, E.~V. Konstantinova, and H. Li,
Thin divisible designs graphs: an interplay between fixed-point free
involutions of $(v,k,\lambda)$-graphs and symmetric weighing matrices,
preprint, 2025.
\href{https://arxiv.org/abs/2512.16653}{arXiv:2512.16653}.

\bibitem{Gritsenko}
O. Gritsenko,
On strongly regular graph with parameters $(65;32;15;16)$,
preprint, 2021.
\href{https://arxiv.org/abs/2102.05432}{arXiv:2102.05432}.

\bibitem{Mathon}
R. Mathon,
Symmetric conference matrices of order $pq^2+1$,
\emph{Canadian Journal of Mathematics} \textbf{30} (1978), no.~2,
321--331.
\href{https://doi.org/10.4153/CJM-1978-029-1}{doi:10.4153/CJM-1978-029-1}.

\bibitem{NetworkX}
A.~A. Hagberg, D.~A. Schult, and P.~J. Swart,
Exploring network structure, dynamics, and function using NetworkX,
\emph{Proceedings of the 7th Python in Science Conference} (2008), 11--15.
\url{https://networkx.org/}.

\bibitem{StrohmerHeath}
T. Strohmer and R.~W. Heath, Jr.,
Grassmannian frames with applications to coding and communication,
\emph{Applied and Computational Harmonic Analysis} \textbf{14}
(2003), no.~3, 257--275.
\href{https://doi.org/10.1016/S1063-5203(03)00023-X}
{doi:10.1016/S1063-5203(03)00023-X}.

\bibitem{Taylor}
D.~E. Taylor,
Regular 2-graphs,
\emph{Proceedings of the London Mathematical Society (3)} \textbf{35}
(1977), 257--274.
\href{https://doi.org/10.1112/plms/s3-35.2.257}{doi:10.1112/plms/s3-35.2.257}.

\bibitem{Waldron}
S. Waldron,
On the construction of equiangular frames from graphs,
\emph{Linear Algebra and its Applications} \textbf{431}
(2009), 2228--2242.
\href{https://doi.org/10.1016/j.laa.2009.07.016}
{doi:10.1016/j.laa.2009.07.016}.

\end{thebibliography}
}
\end{document}